\documentclass[aps,a4paper,preprint,tightenlines,showpacs]{revtex4-1}
\usepackage{lmodern,amsmath,amsfonts,amsthm}
\usepackage{braket,graphicx}

\newcommand{\setN}{\mathbb{N}}
\newcommand{\setC}{\mathbb{C}}
\DeclareMathOperator{\tr}{tr}
\DeclareMathOperator{\id}{id}
\newtheorem{prop}{Proposition}

\theoremstyle{remark}
\newtheorem{example}{Example}

\begin{document}
\title{Many unstable particles from an open quantum systems
  perspective}

\author{Kordian Andrzej Smoli\'nski}
\email{K.A.Smolinski@merlin.phys.uni.lodz.pl}
\affiliation{Department of Theoretical Physics, 
  Faculty of Physics and Applied Computer Sciences \\
  University of \L\'od\'z,  ul. Pomorska 149/153, 90-236 \L\'od\'z, Poland}

\begin{abstract}
  We postulate a master equation, written in the language of creation
  and annihilation operators, as a candidate for unambiguous quantum
  mechanical description of unstable particles.  We have found Kraus
  representation for the evolution driven by this master equation and
  study its properties.  Both Schr\"odinger and Heisenberg picture of
  the system evolution are presented.  We show that the resulting time
  evolution leads to exponential decay law.  Moreover, we analyse
  mixing of particle flavours and we show that it can lead to flavour
  oscillation phenomenon.
\end{abstract}

\maketitle



\section{Introduction}
\label{sec:intro}

\noindent 
One of important difficulties in quantum mechanical description of
unstable particles is irreversibility of time-evolution.  The complete
system consists of decaying particle as well as of decay products.
Only this complete system undergoes unitary evolution, described by
quantum field theory.  However, in many applications of quantum
mechanics (e.g., in analysis of correlations experiments) we would
like to neglect the evolution of decay products and consider the
decaying particles only.  This is usually achieved by introducing a
non-Hermitian Hamiltonian, as it was done in the classical works of
Weisskopf and Wigner \cite{weisskopf30a,weisskopf30b}.  The
non-Hermitian Hamiltonian, however, leads inevitably to the
non-conservation of trace of the density operator of the system.
Although such description gives decreasing probability of detecting
the particle, it does not provide unambiguous way of calculating the
probability of finding the system consisting of a few such particles
in a given state after the measurement (see e.g.\ \cite{durt13} for a
discussion of other ambiguities caused by various approaches to the
description of such a system).  Indeed, one needs to use probability
theory rather than quantum mechanics for this purpose.  Since the
calculations of conditional probabilities are extremely important for
analysis of correlation experiments, especially of those done for
systems of neutral kaons \cite{uchiyama97,bramon99,foadi99a,dalitz01,
  bertlmann01,bertlmann01a,bramon02,bertlmann03,genovese04} or
$B$-mesons \cite{go04,bramon05a,go07:abb}, it would be desirable to
find a quantum mechanical description of decaying particles preserving
unit trace and positivity
\cite{benatti96,benatti97a,benatti97b,benatti98,benatti98a,
  benatti98b:full,benatti01} of the density operator for the system.
Recent papers (e.g.\ \cite{ellis96,bernabeu13}) show that there is
still a great interest in the unambiguous quantum-mechanical
description of neutral kaon system.

This is the point where the theory of dynamical semigroups and open
quantum systems \cite{kossakowski72,gorini76,lindblad76} can be
helpful.  Let us recall (cf.\
\cite{kossakowski98:full,alicki01,breuer02}) that the dynamical
semigroup in the Schr\"odinger picture is a one-parameter family of
linear maps $\Lambda_t^*$, acting on the space of trace class
operators on Hilbert space of the system, preserving for every $t \geq
0$: (i)~positivity, (ii)~trace, (iii)~strong continuity and such that
(iv)~$\Lambda_{t_1}^* \Lambda_{t_2}^* = \Lambda_{t_1 + t_2}^*$ for
every $t_1, t_2 \geq 0$.  These properties can be translated into
Heisenberg picture as requirements for the map $\Lambda_t$ acting on
the space of bounded operators on the Hilbert space of the system,
which for every $t \geq 0$: (i)~preserves the positive cone,
(ii)~leaves the identity operator invariant, (iii)~is continuous on
states in the trace-norm sense, (iv)~is normal and (v)~$\Lambda_{t_1}
\Lambda_{t_2} = \Lambda_{t_1 + t_2}$ for $t_1, t_2 \geq 0$.

The idea that the theory of open quantum system would be useful for
the description of unstable particles appeared quite early
\cite{alicki77,alicki78,weron85} (see also \cite{alicki07:full} for a
review).  Recently, the open quantum system approach was also applied
to the systems of particles with flavour oscillations (like in the
case of neutral kaons) \cite{caban05a,bertlmann06}, and it has been
used successfully in the description of EPR correlations and evolution
of entanglement in $K^0 \overline{K}^0$ system
\cite{caban06a:abb,caban07}.

Here, we follow the approach presented in
\cite{caban05a,caban06a:abb,caban07}.  However, in these works the
considerations were restricted to systems of at most two particles,
and transition from one-particle to two-particle theory was done by
means of tensor product construction.  In the present paper we will
show that it is possible to describe systems with arbitrary number of
particles using the second quantization formalism, which is the most
natural language for system with varying number of particles.
Moreover, such approach would be an advantage if we study the
behaviour of the system from uniformly moving or accelerated frame,
due to the well established transformation properties of annihilation
and creation operators.

The paper is organized as follows.  In Sect.~\ref{sec:schroedinger},
we postulate a master equation in Schr\"odinger picture for a single
kind of free particles, and then find the solution of this equation in
the form of the Kraus representation of the evolution of the density
operator of the system.  The next section is devoted to the Heisenberg
picture of the evolution of the same system.  In
Sect.~\ref{sec:many-types}, we analyse the system of particles of
different types and the flavour oscillation phenomenon.

\section{Schr\"odinger Picture}
\label{sec:schroedinger}

\noindent
In \cite{caban05a} it was shown that the time evolution of a free
unstable scalar particle can be described by a master equation in the
Lindblad--Gorini--Kossakowski--Sudarshan form
\cite{lindblad76,gorini76}:
\begin{equation}
  \label{eq:1}
  \frac{d\rho(t)}{dt} = -i [H, \rho(t)] + \{K, \rho(t)\} 
  + L \rho(t) L^\dagger\,,
\end{equation}
where
\begin{equation}
  \label{eq:2}
  H = m \ket{1} \bra{1}\,, \quad
  L = \sqrt{\Gamma} \ket{0} \bra{1}\,, \quad
  K = -\frac{1}{2} L^\dagger L\,.
\end{equation}
Here $\ket{1}$ denotes the state of presence of the particle and
$\ket{0}$ denotes the state of its absence; $m$ is the mass of the
particle and $\Gamma$ is its decay width.  Despite the fact that the
state $\ket{0}$ is usually called vacuum, it is not the vacuum in the
sense of quantum field theory, but rather in the sense used in
\cite{caban02}, i.e., it is an absence of a particle.  This equation
leads to the probability density of finding the particle evolving
according to the Geiger--Nutall exponential law.

However, the most natural quantum-mechanical description of systems
with variable number of particles is the second quantization
formalism.  For systems governed by \eqref{eq:1} the transition to
second quantization is straightforward, since the operators
\eqref{eq:2} can be interpreted as the vacuum--one-particle sector of
the second quantized operators
\begin{equation}
  \label{eq:3}
  \hat{H} = m \hat{a}^\dagger \hat{a} \equiv m \hat{N}\,, \quad
  \hat{L} = \sqrt{\Gamma} \hat{a}\,,
\end{equation}
where $\hat{a}$ and $\hat{a}^\dagger$ are bosonic annihilation and
creation operators, respectively:
\begin{equation}
  \label{eq:4}
  [\hat{a}, \hat{a}^\dagger] = 1\,, \quad
  \hat{a} \ket{0} = 0\,, \quad
  \ket{n} = \frac{(\hat{a}^\dagger)^n}{\sqrt{n!}} \ket{0}\,;
\end{equation}
vectors $\ket{n}$ form the so-called occupation number basis.  With
this substitution we have
\begin{displaymath}
  \hat{K} = -\frac{1}{2} \Gamma \hat{N}\,.
\end{displaymath}
If we substitute~\eqref{eq:3} into \eqref{eq:1}, we arrive to the
master equation in the form
\begin{subequations}
  \label{eq:5}
  \begin{gather}
    \frac{d}{dt} \Lambda_t^*\rho = \mathcal{L}^*(\Lambda_t^*\rho)\,,
    \intertext{where} %
    \mathcal{L}^*(\rho) = -i m [\hat{N}, \rho] - \frac{\Gamma}{2}
    \{\hat{N}, \rho\} + \Gamma \hat{a} \rho \hat{a}^\dagger\,.
  \end{gather}
\end{subequations}
In the above we can recognize the equation introduced and studied in
\cite{alicki77,alicki78,alicki01,alicki07:full}.  Similar equations
lead to evolution given by quasi-free semigroups, see eg.\ \cite{demoen77,
  evans77,vanheuverzwijn78,evans79,alicki07:full,blanchard07,blanchard09,
  blanchard09:corr,hellmich10}
and are also studied in the context of quantum optics
\cite{olkiewicz08a,olkiewicz08b,ma08,olkiewicz09}.

If we write down explicitly annihilation and creation operators in
occupation number basis, namely $\hat{a}_{kl} = \sqrt{k+1}
\delta_{k+1,l}$, $\hat{a}^\dagger_{kl} = \sqrt{l+1} \delta_{k,l+1}$,
so $\hat{N}_{kl} = k \delta_{kl}$ ($k, l = 0, 1,\ldots$), then we can
view \eqref{eq:5} as the following infinite system of equations for
matrix elements of the density operator $\Lambda_t^* \rho \equiv
\sum_{kl} \rho_{kl}(t) \ket{k}\bra{l}$:
\begin{equation}
  \label{eq:6}
  \frac{d\rho_{kl}(t)}{dt} = \left[i (k - l) m - \tfrac{1}{2} (k +l)
    \Gamma\right] \rho_{kl}(t) 
  + \sqrt{(k + 1) (l + 1)} \Gamma \rho_{k+1,l+1}(t)\,,
\end{equation}
for $k, l = 0, 1,\ldots$\,.  Thus we get infinite, in principle,
system of linear differential equations of first order.

Notice, that the system~\eqref{eq:6} seems to be highly non-trivial
--- the solution for $\rho_{kl}(t)$ depends on a solution for
$\rho_{k+1,l+1}(t)$, what apparently leads to infinite chain of
dependencies.  What makes the system~\eqref{eq:6} solvable is the
proper choice of initial conditions.  Indeed, for every reasonable
initial physical state, the number of particles must be finite, so all
matrix elements of $\rho(0)$ corresponding to higher number of
particles must vanish.  Mathematically, it means that there exist
indices $r$ and $s$, such that
\begin{equation}
  \label{eq:7}
  \rho_{kl}(0) = 0 \text{ for }  k > r \text{ and } l > s\,.
\end{equation}
One can easily check that the system~\eqref{eq:6} with initial
condition~\eqref{eq:7} gives us the well posed Cauchy's problem.

Now, instead directly solving the equation~\eqref{eq:5} for some
interesting choices of initial state we concentrate on finding and
studying the Kraus representation \cite{kraus83} of the evolution of
the system.

Although Kraus operators for the evolution of the density operator
governed by the master equation~\eqref{eq:5} was found
in~\cite{liu04}, here we give its another formulation, now written in
terms of annihilation/creation operators.  It can be easily checked
that these two choices of Kraus operators coincide up to the phase
factors.  Despite this, we give the formal proof that proposed Kraus
operators lead to the evolution of the system undergoing the
equation~\eqref{eq:5}, since we will employ the technique used in the
proof later on.
\begin{prop}
  \label{prop:1}
  If for $k = 0, 1,\ldots$:
  \begin{equation}
    \label{eq:8}
    E_k(t) = \frac{1}{\sqrt{k!}} e^{-i \hat{M} t} 
    \left(\sqrt{1 - e^{-\Gamma t}} \hat{a}\right)^k\,, 
  \end{equation}
  where
  \begin{math}
    \hat{M} = (m - \frac{i}{2} \Gamma) \hat{N}\,,
  \end{math}
  then
  \begin{equation}
    \label{eq:9}
    \Lambda_t^* \rho = \sum_{k=0}^\infty E_k(t) \rho E_k^\dagger(t)
  \end{equation}
  is the solution of the master equation~\eqref{eq:5}, where $\rho$ is
  the density operator given at initial time $t = 0$.
\end{prop}
\begin{proof}
  First, let us note that for any $k$
  \begin{equation}
    \label{eq:10}
    E_k(t) \hat{a} = e^{(i m + \frac{1}{2} \Gamma) t} \hat{a} E_k(t)\,,
  \end{equation}
  what follows immediately from canonical commutation relations.
  Using~\eqref{eq:10} we can show by straightforward calculation that
  \begin{subequations}
    \label{eq:11}
    \begin{align}
      \label{eq:11a}
      \frac{dE_0(t)}{dt} &= -i \hat{M}
      E_0(t)\,, \\
      \label{eq:11b}
      \frac{dE_k(t)}{dt} &= -i \hat{M} E_k(t) + \frac{\sqrt{k} \Gamma
        e^{(i m + \frac{1}{2} \Gamma) t}}{2 \sqrt{1 - e^{-\Gamma t}}}
      \hat{a} E_{k-1}(t)
    \end{align}
    (for $k = 1, 2,\ldots$).
  \end{subequations}
  
  Next, one can easily check that the following recurrence relations
  hold for $k = 1, 2,\ldots$
  \begin{equation}
    \label{eq:12}
    E_k(t) = \frac{\sqrt{1 - e^{-\Gamma t}}}{\sqrt{k}} 
    e^{(i m + \frac{1}{2} \Gamma) t} \hat{a} E_{k-1}(t)\,.
  \end{equation}
  Combining~\eqref{eq:12} and~\eqref{eq:11b} we can
  write~\eqref{eq:11} in the form
  \begin{equation}
    \label{eq:13}
    \frac{dE_k(t)}{dt} = \left(-i \hat{M} 
    + \frac{k}{2} \frac{\Gamma e^{-\Gamma t}}{1 - e^{-\Gamma t}}\right) E_k(t)\,,
  \quad k = 0, 1,\ldots
  \end{equation}
  Now, let us compute the time derivative of the density operator
  $\Lambda_t^*\rho$ given by~\eqref{eq:9}:
  \begin{multline}
    \label{eq:14}
    \frac{d}{dt}\Lambda_t^*\rho = \sum_{k=0}^\infty
    \left(\frac{dE_k(t)}{dt} \rho E_k^\dagger(t) + E_k(t) \rho
      \frac{dE_k^\dagger(t)}{dt}\right) \\
    = -i \left[\hat{M} (\Lambda_t^*\rho) - (\Lambda_t^*\rho)
      \hat{M}^\dagger\right] 
    + \frac{\Gamma e^{-\Gamma t}}{1 - e^{-\Gamma t}} \sum_{k=1}^\infty
    k E_k(t) \rho E_k^\dagger(t)\,.
  \end{multline}
  Taking into account~\eqref{eq:12} the last term in~\eqref{eq:14} can be
  written as
  \begin{equation}
    \label{eq:15}
    \frac{\Gamma e^{-\Gamma t}}{1 - e^{-\Gamma t}} \sum_{k=1}^\infty k
    E_k(t) \rho E_k^\dagger(t) 
    = \Gamma \hat{a} \sum_{k=1}^\infty E_{k-1}(t) \rho
    E_{k-1}^\dagger(t) \hat{a}^\dagger = \Gamma \hat{a}
    (\Lambda_t^*\rho) \hat{a}^\dagger\,.
  \end{equation}
  Thus, the density operator obeys the master equation~\eqref{eq:5}.

  To complete the proof, we have to show that at the time $t = 0$ the
  density operator $\Lambda_t^*\rho$ given by~\eqref{eq:9} is $\rho$.
  This is trivial point, because obviously $E_0(0) = \id$ and $E_k(0)
  = 0$ for $k = 1, 2, \ldots$
\end{proof}
\noindent
It is easy to see that after writing out annihilation operators in
occupation number basis the Kraus operators~\eqref{eq:8} differ only
by phase factors from those found in \cite{chuang96,liu04}.  These
phase factors become important when you try to study a flavour
oscillation phenomenon (see Sect.~\ref{sec:many-types}).

\begin{prop}
  \label{prop:2}
  If $E_k(t)$ are given by~\eqref{eq:8}, then
  \begin{equation}
    \label{eq:16}
    \sum_{k=0}^\infty E_k^\dagger(t) E_k(t) = \id\,.
  \end{equation}
\end{prop}
\begin{proof}
  We start with the observation that for any element of occupation
  number basis $\ket{n}$ and any non-negative integer $k$
  \begin{subequations}
    \label{eq:17}
    \begin{align}
      \label{eq:17a}
      \hat{a}^k \ket{n} &= 
      \begin{cases}
        \displaystyle
        \sqrt{\frac{n!}{(n-k)!}} \ket{n-k}\,, & n \geq k\,, \\
        0\,, & n < k\,,
      \end{cases}
      \intertext{and} %
      (\hat{a}^\dagger)^k \ket{n} &= \sqrt{\frac{(n+k)!}{n!}}
      \ket{n+k}\,.
    \end{align}
  \end{subequations}
  Since
  \begin{displaymath}
    E_k^\dagger(t) E_k(t) = \frac{1}{k!} \left(1 - e^{-\Gamma t}\right)^k 
    (\hat{a}^\dagger)^k e^{-\Gamma \hat{N} t} \hat{a}^k\,,
  \end{displaymath}
  then, from~\eqref{eq:17}, for any element of the occupation number
  basis
  \begin{equation}
    \label{eq:18}
    E_k^\dagger(t) E_k(t) \ket{n} =
    \binom{n}{k} \left(1 - e^{-\Gamma t}\right)^k e^{-(n-k) \Gamma
      t} \ket{n}\,,
  \end{equation}
  when $n \geq k$, and $E_k^\dagger(t) E_k(t) \ket{n} = 0$, when $n <
  k$.  Thus,
  \begin{equation}
    \label{eq:19}
    \sum_{k=0}^\infty E_k^\dagger(t) E_k(t) \ket{n} = \sum_{k=0}^n
    E_k^\dagger(t) E_k(t) \ket{n} 
    = \sum_{k=0}^n \binom{n}{k} \left(1 - e^{-\Gamma t}\right)^k
    e^{-(n-k) \Gamma t} \ket{n} = \ket{n}\,.
  \end{equation}
  Since $\sum_{k=0}^\infty E_k^\dagger(t) E_k(t)$ acts as identity on
  any element of the basis, it must be the identity operator.
\end{proof}

\begin{prop}
  \label{prop:3}
  Vacuum state $\ket{0} \bra{0}$ is stable under the evolution given
  by~\eqref{eq:8} and~\eqref{eq:9}.  Moreover, $\lim_{t \to
    \infty}\Lambda_t^* \rho = \ket{0} \bra{0}$ for any density
  operator $\rho$.
\end{prop}
\begin{proof}
  Indeed, $E_0(t) \ket{0} = \ket{0}$ and $E_k(t) \ket{0} = 0$ for $k =
  1, 2, \ldots$, so the density operator $\ket{0} \bra{0}$ is stable
  during the time evolution.

  For the proof of the second statement, one can find that
  \begin{equation}
    \label{eq:20}
    E_k(t) \ket{n} =
    \sqrt{\binom{n}{k}} e^{-(i m + \frac{1}{2} \Gamma) (n - k) t} 
    \left(\sqrt{1 - e^{-\Gamma t}}\right)^k \ket{n - k}\,,
  \end{equation}
  when $ n \geq k$, and vanishes otherwise.  Next, observe that for
  $\Gamma > 0$
  \begin{equation}
    \label{eq:21}
    \lim_{t \to \infty} e^{-\frac{1}{2} \Gamma (n - k) t} 
    \left(\sqrt{1 - e^{-\Gamma t}}\right)^k =
    \begin{cases}
      0\,, & n > k\,, \\
      1\,, & n = k\,, \\
      \infty\,, & n < k\,,
    \end{cases}
  \end{equation}
  so,
  \begin{equation}
    \label{eq:22}
    \lim_{t \to \infty} E_k(t) \ket{n} =
    \begin{cases}
      0\,, & n \neq k\,, \\
      \ket{0}\,, & n = k\,,
    \end{cases}
  \end{equation}
  and, consequently,
  \begin{equation}
    \label{eq:23}
    \lim_{t \to \infty} E_k(t) \ket{n} \bra{n'} E_k^\dagger(t) =
    \begin{cases}
      0\,, & n \neq n'\,, \\
      \delta_{nk} \ket{0} \bra{0}\,, & n = n'\,.
    \end{cases}
  \end{equation}
  Thus, for any density operator $\rho$, $\lim_{t \to \infty}
  \Lambda_t^* \rho = \tr(\rho) \ket{0} \bra{0} = \ket{0} \bra{0}$.
\end{proof}

Explicit solutions of \eqref{eq:1} with operators \eqref{eq:3} can be
deduced from the relation
\begin{multline}
  \label{eq:24}
  \Lambda_t^* \ket{n} \bra{n'} = \sum_{k=0}^{\min\{n, n'\}}
  \sqrt{\binom{n}{k} \binom{n'}{k}} e^{-i m (n - n') t}\\
  \times e^{-\frac{1}{2} \Gamma (n + n' - 2 k) t} \left(1 - e^{-\Gamma
      t}\right)^k \ket{n - k} \bra{n' - k}\,.
\end{multline}
If we impose the superselection rule which forbids the superpositions
of states with different number of particles, then the density
operator for a system consisting of at most $n$ particles is of the
form
\begin{displaymath}
  \Lambda_t^*\rho = \sum_{k=0}^n p_k(t) \ket{k} \bra{k}\,, \quad
  \sum_{k=0}^n p_k(t) = 1\,.
\end{displaymath}
Thus, it is enough to solve the equation~\eqref{eq:1} for an initial
state of the form $\rho = \ket{n} \bra{n}$ for $n$ being some
non-negative integer (arbitrary, but finite), because any density
operator for the initial state is a linear combination of such states.

If the system is initially in the $n$-particle pure state, $\rho =
\ket{n} \bra{n}$, then the solution of the equation~\eqref{eq:1} is
\begin{displaymath}
  \Lambda_t^*\rho = \sum_{k=0}^n \binom{n}{k} e^{-(n-k) \Gamma t} 
  \left(1 - e^{-\Gamma t}\right)^k 
  \ket{n-k} \bra{n-k}\,.
\end{displaymath}
The average number of particles changes in time as
\begin{displaymath}
  \left<N(t)\right> = \tr\left[(\Lambda_t^*\rho) \hat{N}\right] 
  = n e^{-\Gamma t}\,.
\end{displaymath}
We have thus recovered the Geiger--Nutall exponential decay law.  It
is worth noting that the probability that at a time $t$ one finds
exactly $k$ particles from initally $n$ ones, has a binomial
distribution $B(n, e^{-\Gamma t})$ with probability $e^{-\Gamma t}$ of
finding a single particle, as it can be expected.

\section{Heisenberg Picture}
\label{sec:heisenberg}
\noindent
In Heisenberg picture, master equation for the evolution of an
observable $\Omega$ is of the form
\begin{subequations}
  \label{eq:27}
  \begin{gather}
    \label{eq:27a}
    \frac{d}{dt} \Lambda_t \Omega = \mathcal{L}(\Lambda_t \Omega)\,, 
    \intertext{where} 
    \label{eq:27b}
    \mathcal{L}(\Omega) = i [H, \Omega] + \frac{1}{2}
    \left\{\left[L^\dagger, \Omega\right] L + L^\dagger \left[\Omega,
        L\right]\right\}\,.
  \end{gather}
\end{subequations}
Note that 
\begin{displaymath}
  \frac{1}{2} \{[L^\dagger, \Omega] L +
  L^\dagger [\Omega, L]\} = \{K, \Omega\} + L^\dagger
  \Omega L\,,
\end{displaymath}
but the form used in~\eqref{eq:27b} is usually more convenient when
performing calculations in the Heisenberg picture involving creation
and annihilation operators.

Having a family of Kraus operators~\eqref{eq:8}, the evolution of
observable $\Omega$ can be written as the series
\begin{equation}
  \label{eq:25}
  \Lambda_t \Omega = \sum_{k=0}^\infty E_k^\dagger(t) \Omega E_k(t)\,.
\end{equation}
This representation is especially useful if we can find the
decomposition of the observable into its matrix elements in occupation
basis:
\begin{equation}
  \label{eq:32}
  \Omega = \sum_{n,n'} \omega_{n,n'} \ket{n} \bra{n'}\,.
\end{equation}
\begin{multline}
  \label{eq:26}
  \Lambda_t \ket{n} \bra{n'} = \sum_{k=0}^\infty \sqrt{\binom{n +
      k}{n} \binom{n' + k}{n'}} e^{i m (n - n') t} \\
  \times e^{-\frac{1}{2} \Gamma (n + n') t} \left(1 - e^{-\Gamma
      t}\right)^k \ket{n + k} \bra{n' + k}\,.
\end{multline}
Using projectors onto $n$-particle states, $\hat{\Pi}_n \equiv \ket{n}
\bra{n}$, the last equation can be rewritten in a more convenient form
\begin{equation}
  \label{eq:38}
  \Lambda_t \hat{\Pi}_n = \frac{1}{\left(e^{\Gamma t} - 1\right)^n} 
  \sum_{k=n}^\infty \binom{k}{n} \left(1 - e^{-\Gamma t}\right)^k 
  \hat{\Pi}_k\,.
\end{equation}

\begin{prop}
  \label{prop:4}
  $\lim_{t \to \infty} \Lambda_t \hat{\Pi}_0 = \id$.
\end{prop}
\begin{proof}
  From~\eqref{eq:38} it follows that
  \begin{displaymath}
    \Lambda_t \hat{\Pi}_0 = \sum_{k=0}^\infty 
    \left(1 - e^{-\Gamma t}\right)^k \hat{\Pi}_k\,,
  \end{displaymath}
  so $\lim_{t \to \infty} \Lambda_t \hat{\Pi}_0 =
  \displaystyle\sum_{k=0}^\infty \hat{\Pi}_k \equiv \id$.
\end{proof}
\noindent
Physically, Proposition~\ref{prop:4} tells us that after substantially
long (mathematically infinite) period of time, the probability of
finding vacuum reaches one, irrespectively of the state of the system.
In other words, at infinite time all the Fock spaces collapse to the
vacuum subspace.

The evolution of creation and annihilation operators can be easily
find with use of relation~\eqref{eq:10}:
\begin{subequations}
  \label{eq:34}
  \begin{gather}
    \Lambda_t \hat{a} = e^{-(i m + \frac{1}{2} \Gamma) t} \hat{a}\,, \\
    \Lambda_t \hat{a}^\dagger = e^{(i m - \frac{1}{2} \Gamma) t} \hat{a}^\dagger\,.
  \end{gather}
\end{subequations}
Moreover, it is easy to check that in this case $\Lambda_t \hat{N} =
\Lambda_t \hat{a}^\dagger \Lambda_t \hat{a}$ (what, in general, does
not hold).  Indeed, using~\eqref{eq:10} we have
\begin{equation}
  \label{eq:48}
  \Lambda_t \hat{N} = \sum_{k=0}^\infty E_k^\dagger(t) \hat{a}^\dagger
  \hat{a} E_k(t)
  = \left(\sum_{k=0}^\infty E_k^\dagger(t) \hat{a}^\dagger
    E_k(t)\right) e^{-i (m + \frac{1}{2} \Gamma) t} \hat{a} =
  \Lambda_t \hat{a}^\dagger \Lambda_t \hat{a}\,.
\end{equation}
Consequently, the evolution of the particle number observable is
\begin{equation}
  \label{eq:44}
  \Lambda_t \hat{N} = e^{-\Gamma t} \hat{N}\,;
\end{equation}
we can get this result by solving~\eqref{eq:27a} for $\hat{N}$, too.

It is easy to find the mean number of particles for a given state with
help of~\eqref{eq:44}.  Here, we consider two examples: the pure state
of exactly $n$ particles and a coherent state with given mean number
of particles $\bar{n}$.

\begin{example}
  If the system is in the pure state of $n$ particles, then the mean
  number of particles is simply
  \begin{equation}
    \label{eq:37}
    \left<N(t)\right> 
    = n e^{-\Gamma t}\,.
  \end{equation}
  Thus we get the exponential decay law again.  Time evolution of the
  probability of finding exactly $k$ particles follows
  from~\eqref{eq:38} and reads
  \begin{equation}
    \label{eq:45}
    p_n(k, t) 
    = \binom{n}{k} e^{-k \Gamma t} \left(1 - e^{-\Gamma t}\right)^{n-k}\,,
  \end{equation}
  i.e., it is given by the binomial distribution $B(n,e^{-\Gamma t})$.
\end{example}

\begin{example}
  Let us assume that the system is in a coherent state $\ket{\alpha}$,
  \begin{displaymath}
    a \ket{\alpha} = \alpha \ket{\alpha}\,,
  \end{displaymath}
  $\alpha \in \setC$, i.e.
  \begin{equation}
    \label{eq:39}
    \ket{\alpha} = e^{-\frac{|\alpha|^2}{2}} \sum_{k=0}^\infty 
    \frac{\alpha^k}{\sqrt{k!}} \ket{k}\,,
  \end{equation}
  then
  \begin{equation}
    \label{eq:40}
    \left<N(t)\right> 
    = \bar{n} e^{-\Gamma t}\,,
  \end{equation}
  where $\bar{n} \equiv |\alpha|^2$ is the mean number of particles in
  the coherent state $\ket{\alpha}$.  Probability of finding exactly $k$
  particles evolves in time according to
  \begin{equation}
    \label{eq:46}
    p_{\bar{n}}(k, t) 
    = \frac{1}{k!} \left(\bar{n} e^{-\Gamma t}\right)^k e^{-\bar{n} e^{-\Gamma t}}\,,
  \end{equation}
  which is the Poisson distribution $P(\bar{n} e^{-\Gamma t})$.

  Let us note that, if we consider the state being a mixture of
  $k$-particle states with probability that $k$-particle state occurs
  given by the Poisson distribution with mean number of particles
  $\bar{n}$, i.e.,
  \begin{equation}
    \label{eq:50}
    \rho = \sum_{k=0}^\infty \frac{e^{-\bar{n}} \bar{n}^k}{k!} 
    \ket{k} \bra{k}\,,
  \end{equation}
  then the mean number of particles in this state and probability of
  finding exactly $k$ particles are given by the
  formulae~\eqref{eq:40} and~\eqref{eq:46}, respectively (despite the
  fact, that in this case we must find the traces of the product of
  observables with the density operator).
\end{example}

\section{Particles of Different Types}
\label{sec:many-types}
\noindent
Let us consider a system of particles of $r$ different types (or
carrying a quantum number with $r$ possible values), each type with
mass $m_j$ and width $\Gamma_j$ for $j = 1,\ldots, r$.  For such a
system we have
\begin{equation}
  \label{eq:60}
  [\hat{a}_j, \hat{a}_k]_\mp = 0\,, \quad
  [\hat{a}_j, \hat{a}_k^\dagger]_\mp = \delta_{jk}\,,
\end{equation}
for $j, k = 1,\ldots, r$, where $[\cdot,\cdot]_\mp$ denotes
commutator/anti-commutator, respectively, and anti-commutators apply
only if both $j$\textsuperscript{th}- and
$k$\textsuperscript{th}-particles are fermions.  The states spanning
the occupation number representation are generated from the vacuum
state via the formula
\begin{equation}
  \label{eq:61}
  \ket{n_1, n_2,\ldots, n_r} = \frac{(\hat{a}_1^\dagger)^{n_1} 
    (\hat{a}_2^\dagger)^{n_2} \cdots (\hat{a}_r^\dagger)^{n_r}}
  {\sqrt{n_1! n_2! \cdots n_r!}} \ket{0}
\end{equation}
where we identify $\ket{0,0,\ldots, 0} \equiv \ket{0}$.

The master equation for the system takes the following forms
\begin{subequations}
  \label{eq:28}
  \begin{align}
    \frac{d}{dt} \Lambda_t^* \rho &= -i [\hat{H}, \Lambda_t^* \rho] 
    + \{\hat{K}, \Lambda_t^* \rho\}
    + \sum_{j=1}^r \hat{L}_j (\Lambda_t^* \rho) \hat{L}_j^\dagger\,, \\
    \label{eq:28a}
    \frac{d}{dt} \Lambda_t \Omega &= i [\hat{H}, \Lambda_t \Omega]
    + \sum_{j=1}^r \left\{[\hat{L}_j^\dagger, \Lambda_t \Omega]
      \hat{L}_j + \hat{L}_j^\dagger [\Lambda_t \Omega,
      \hat{L}_j]\right\}\,,
  \end{align}
\end{subequations}
in the Schr\"odinger and Heisenberg picture, respectively, where
$\hat{H}$ is the Hamiltonian of the system and
\begin{equation}
  \label{eq:29}
  \hat{L}_j = \sqrt{\Gamma_j} \hat{a}_j\,, \quad
  \hat{K} = -\frac{1}{2} \sum_{j=1}^r \hat{L}_j^\dagger L_j\,, \quad
  \hat{M} = \hat{H} + i \hat{K}\,.
\end{equation}
If $[\hat{M}, \hat{a}_j] = -(m_j - \frac{i}{2} \Gamma_j) \hat{a}_j$
for $j = 1,\ldots, r$, then we can easily construct the Kraus operators
solving~\eqref{eq:28}
\begin{equation}
  \label{eq:30}
  E_k(t) = e^{-i \hat{M} t}
  \prod_{\substack{k_1,\ldots, k_r\\k_1 + \cdots + k_r = k}} 
  \frac{\left(\sqrt{1 - e^{-\Gamma_j t}} 
      \hat{a}_j\right)^{k_j}}{\sqrt{k_j!}} \,,
\end{equation}
where the product is taken over all possible partitions of $k$ into
exactly $r$ addends, such that $k_1 + k_2 + \cdots + k_r = k$, where
\begin{subequations}
  \label{eq:31}
  \begin{align}
    \label{eq:31a}
    k_j &\in \setN_0\,, & & j^{\text{th}}\text{-particles are
      bosons}\,, \\
    k_j &\in \{0, 1\}\,, & & j^{\text{th}}\text{-particle is a
      fermion}\,.
  \end{align}
\end{subequations}
for $j = 1,\ldots, r$.

\begin{example}
  Let us consider the evolution of a system of two flavour particles
  (e.g., particles and their anti-particles).  We denote the creation
  operators for these particles by $\hat{a}_1^\dagger$ and
  $\hat{a}_2^\dagger$.  The basis for the system is built up from the
  states of the form
  \begin{equation}
    \label{eq:33}
    \ket{n_1,n_2} = \frac{(\hat{a}_1^\dagger)^{n_1} (\hat{a}_2^\dagger)^{n_2}} 
    {\sqrt{n_1! n_2!}} \ket{0}\,.
  \end{equation}
  Let these states be the common eigenstates of two observables,
  $\hat{N} = \hat{a}_1^\dagger \hat{a}_1 + \hat{a}_2^\dagger \hat{a}_2$
  (number of particles) and $\hat{S} = \hat{a}_1^\dagger \hat{a}_1 -
  \hat{a}_2^\dagger \hat{a}_2$ (say strangeness or lepton number), i.e.,
  \begin{align*}
    \hat{N} \ket{n_1,n_2} &= (n_1 + n_2) \ket{n_1,n_2}\,, \\
    \hat{S} \ket{n_1,n_2} &= (n_1 - n_2) \ket{n_1,n_2}\,.
  \end{align*}
  If the states~\eqref{eq:33} are not eigenstates of the time evolution
  the phenomenon known as the flavour oscillation may occur.

  To describe such a situation, let us assume that the Hamiltonian and
  Lindblad operators for the system are of the form
  \begin{subequations}
    \label{eq:43}
    \begin{align}
      \hat{H} &= m_1 \hat{c}_1^\dagger \hat{c}_1
      + m_2 \hat{c}_2^\dagger \hat{c}_2\,, \\
      \hat{L}_1 &= \sqrt{\Gamma_1} \hat{c}_1\,, \\
      \hat{L}_2 &= \sqrt{\Gamma_2} \hat{c}_2\,,
    \end{align}
  \end{subequations}
  where $\hat{c}_1^\dagger, \hat{c}_2^\dagger$ are connected with
  $\hat{a}_1^\dagger, \hat{a}_2^\dagger$ by unitary transformation:
  \begin{subequations}
    \label{eq:35}
    \begin{align}
      \hat{c}_1^\dagger &= e^{i \chi} \left(e^{i (\phi + \psi)/2}
        \cos\frac{\theta}{2} \hat{a}_1^\dagger + e^{-i (\phi - \psi)/2}
        \sin\frac{\theta}{2} \hat{a}_2^\dagger\right)\,, \\
      \hat{c}_2^\dagger &= e^{i \chi} \left(-e^{i (\phi - \psi)/2}
        \sin\frac{\theta}{2} \hat{a}_1^\dagger + e^{-i (\phi + \psi)/2}
        \cos\frac{\theta}{2} \hat{a}_2^\dagger\right)\,.
    \end{align}
  \end{subequations}

  Since $\hat{M} = \left(m_1 - \frac{i}{2} \Gamma_1\right)
  \hat{c}_1^\dagger \hat{c}_1 + \left(m_2 - \frac{i}{2}
    \Gamma_2\right) \hat{c}_2^\dagger \hat{c}_2$, we can easily find
  the evolution of $\hat{c}_j$:
  \begin{equation}
    \label{eq:36}
    \Lambda_t \hat{c}_j = e^{-\left(i m_j + \frac{1}{2} \Gamma_j\right) t} 
    \hat{c}_j\,, \quad j = 1, 2\,.
  \end{equation}
  Using~\eqref{eq:35} we get the evolution of $\hat{a}_j$:
  \begin{subequations}
    \label{eq:41}
    \begin{align}
      \Lambda_t \hat{a}_1 &= \frac{1}{2} e^{-\left(i m_1 + \frac{1}{2}
          \Gamma_1\right) t} \left[\hat{a}_1 (1 + \cos\theta) +
        \hat{a}_2 e^{i \phi} \sin\theta\right] \notag \\
      &\quad + \frac{1}{2} e^{-\left(i m_2 + \frac{1}{2}
          \Gamma_2\right) t} \left[\hat{a}_1 (1 - \cos\theta) -
        \hat{a}_2 e^{i \phi} \sin\theta\right]\,,\\
      \Lambda_t \hat{a}_2 &= \frac{1}{2} e^{-\left(i m_1 + \frac{1}{2}
          \Gamma_1\right) t} \left[\hat{a}_2 (1 + \cos\theta) -
        \hat{a}_1 e^{-i \phi} \sin\theta\right] \notag \\
      &\quad + \frac{1}{2} e^{-\left(i m_2 + \frac{1}{2}
          \Gamma_2\right) t} \left[\hat{a}_2 (1 - \cos\theta) +
        \hat{a}_1 e^{-i \phi} \sin\theta\right]\, 
    \end{align}
  \end{subequations}
  The time evolution of the observables can be obtained either by
  solving \eqref{eq:28a} or directly from relations~\eqref{eq:41},
  using argumentation analogous to~\eqref{eq:48}.  For example, for
  the number of particles we get
  \begin{equation}
    \label{eq:42}
    \Lambda_t \hat{N} = \frac{e^{-\Gamma_1 t} + e^{-\Gamma_2 t}}{2}
    \hat{N} 
    + \frac{e^{-\Gamma_1 t} - e^{-\Gamma_2 t}}{2} [\hat{S} \cos\theta
    + \hat{Q}_+ \sin\theta]\,,
  \end{equation}
  where $\hat{Q}_+ = \hat{a}_1^\dagger \hat{a}_2 e^{i \phi} +
  \hat{a}_2^\dagger \hat{a}_1 e^{-i \phi}$, so the mean value in the
  state $\ket{n_1, n_2}$ is the following
  \begin{equation}
    \label{eq:49}
    \left<N(t)\right> = \frac{e^{-\Gamma_1 t} + e^{-\Gamma_2 t}}{2}
    (n_1 + n_2) 
    + \frac{e^{-\Gamma_1 t} - e^{-\Gamma_2 t}}{2} (n_1 - n_2)
    \cos\theta
  \end{equation}
  and is depicted in the Fig.~\ref{fig:mixing-N}.
  \begin{figure}
    \centering
    \includegraphics[width=\columnwidth]{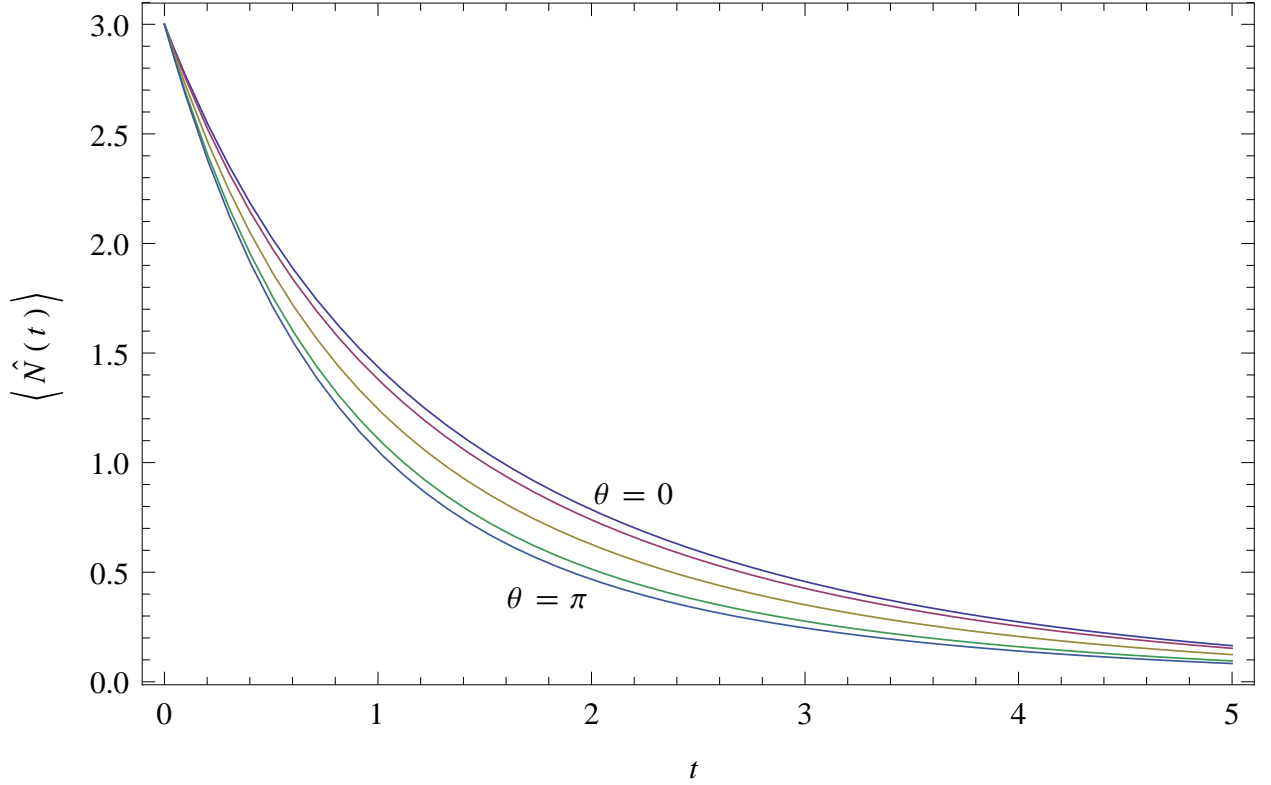}
    \caption{Number of particles for system with two flavours for mixing
      angles $\theta = 0, \frac{\pi}{4}, \frac{\pi}{2}, \frac{3\pi}{4},
      \pi$ (from right to left) with $n_1 = 2$ and $n_2 = 1$, and
      $\Gamma_1 < \Gamma_2$ (time unit is $\tau = 1/\Gamma$).}
    \label{fig:mixing-N}
  \end{figure}

  Similarly, for the strangeness (or lepton number) we get
  \begin{multline}
    \label{eq:47}
    \Lambda_t \hat{S} = \frac{e^{-\Gamma_1 t} - e^{-\Gamma_2 t}}{2}
    \hat{N} \cos\theta + e^{-\Gamma t} \sin(\Delta m t) \hat{Q}_-
    \sin\theta\\
    + \left[\frac{e^{-\Gamma_1 t} + e^{-\Gamma_2 t}}{2} \cos^2\theta +
      e^{-\Gamma t} \cos(\Delta m t) \sin^2\theta\right] \hat{S}\\
    + \left[\frac{e^{-\Gamma_1 t} + e^{-\Gamma_2 t}}{2} - e^{-\Gamma
        t} \cos(\Delta m t)\right] \hat{Q}_+ \sin\theta \cos\theta\,,
  \end{multline}
  where $\hat{Q}_- = i (\hat{a}_1^\dagger \hat{a}_2 e^{i \phi} -
  \hat{a}_2^\dagger \hat{a}_1 e^{-i \phi})$, $\Gamma = \frac{1}{2}
  (\Gamma_1 + \Gamma_2)$ and $\Delta m = m_2 - m_1$.  The mean value
  of this observable in the state $\ket{n_1, n_2}$ is
  \begin{multline}
    \label{eq:51}
    \left<S(t)\right> = \frac{e^{-\Gamma_1 t} - e^{-\Gamma_2 t}}{2}
      (n_1 + n_2) \cos\theta\\
      + \left[\frac{e^{-\Gamma_1 t} + e^{-\Gamma_2 t}}{2} \cos^2\theta
          + e^{-\Gamma t} \cos(\Delta m t) \sin^2\theta\right] (n_1 -
        n_2)
  \end{multline}
  and is shown in the Fig.~\ref{fig:mixing-S}.
  \begin{figure}
    \centering
    \includegraphics[width=\columnwidth]{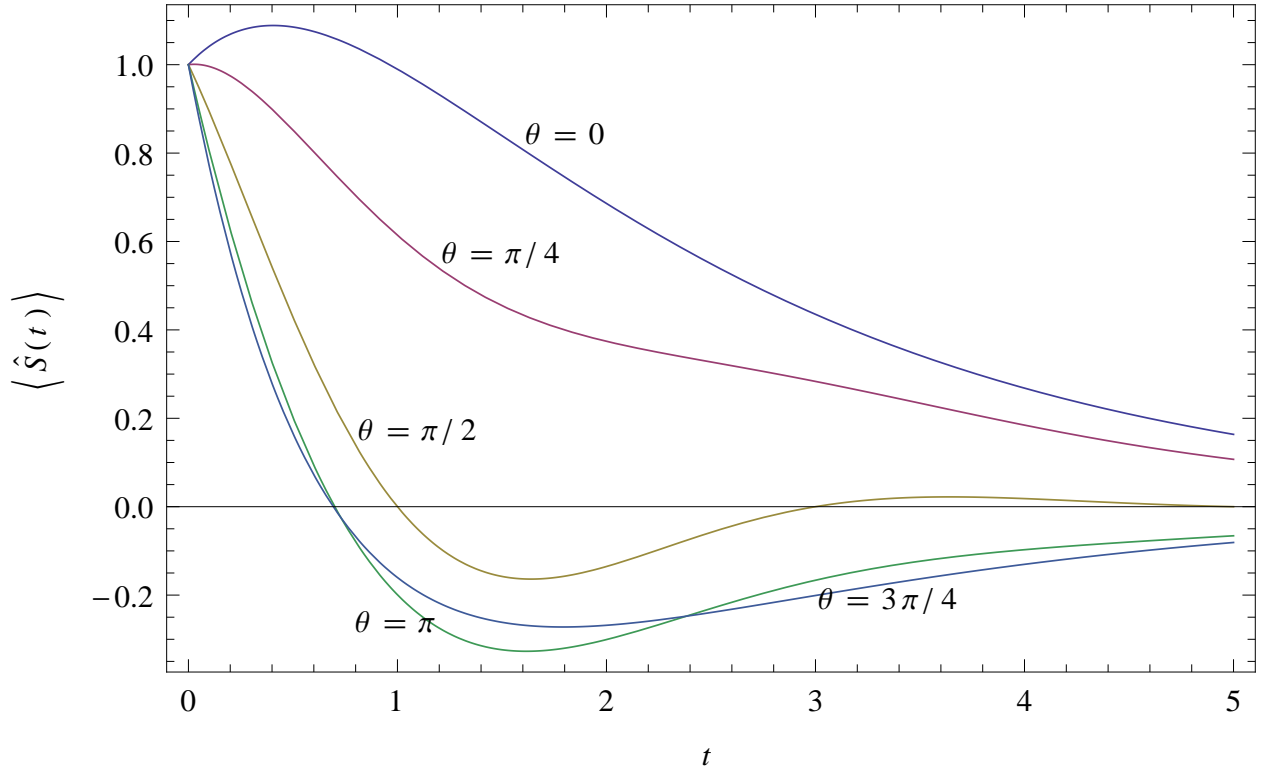}
    \caption{Strangeness of system with two flavours for different values
      of mixing angle with $n_1 = 2$ and $n_2 = 1$, and $\Gamma_1 <
      \Gamma_2$ (time unit is $\tau = 1/\Gamma$).}
    \label{fig:mixing-S}
  \end{figure}

  Let us pay our attention on the two extreme cases: $\theta = 0$ (no
  flavour mixing) and $\theta = \frac{\pi}{2}$ (maximal mixing).  For
  $\theta = 0$ we have $\hat{c}_i = \hat{a}_i$, so the time evolution of
  the observables is
  \begin{subequations}
    \label{eq:52}
    \begin{align}
      \Lambda_t\hat{N} &= \frac{1}{2} e^{-\Gamma_1 t} (\hat{N} + \hat{S}) 
      + \frac{1}{2} e^{-\Gamma_2 t} (\hat{N} - \hat{S})\,, \\
      \Lambda_t\hat{S} &= \frac{1}{2} e^{-\Gamma_1 t} (\hat{S} + \hat{N}) 
      + \frac{1}{2} e^{-\Gamma_2 t} (\hat{S} - \hat{N})\,.
    \end{align}
  \end{subequations}
  Their mean values in the state $\ket{n_1, n_2}$ are
  \begin{subequations}
    \label{eq:53}
    \begin{align}
      \left<N(t)\right> &=  e^{-\Gamma_1 t} n_1 + e^{-\Gamma_2 t} n_2\,, \\
      \left<S(t)\right> &=  e^{-\Gamma_1 t} n_1 - e^{-\Gamma_2 t} n_2\,.
    \end{align}
  \end{subequations}
  For $\theta = \frac{\pi}{2}$, $\phi = 2 \pi$, $\psi = \pi$ and $\chi
  = \frac{3 \pi}{2}$ we have
  \begin{align*}
    \hat{c}_1^\dagger &= \frac{1}{\sqrt{2}} (\hat{a}_1^\dagger +
    \hat{a}_2^\dagger)\,, \\
    \hat{c}_2^\dagger &= \frac{1}{\sqrt{2}} (\hat{a}_1^\dagger -
    \hat{a}_2^\dagger)\,,
  \end{align*}
  and the time evolution of the observables is given by
  \begin{subequations}
    \label{eq:55}
    \begin{align}
      \Lambda_t\hat{N} &= \frac{1}{2} \left(e^{-\Gamma_1 t} +
        e^{-\Gamma_2 t}\right) \hat{N} + \frac{1}{2}
      \left(e^{-\Gamma_1 t} - e^{-\Gamma_2 t}\right) \hat{Q}_+\,, \\
      \Lambda_t\hat{S} &= e^{-\Gamma t} \cos(\Delta m t) \hat{S} +
      e^{-\Gamma t} \sin(\Delta m t) \hat{Q}_-\,.
    \end{align}
  \end{subequations}
  The mean values of these observables in the state $\ket{n_1, n_2}$
  are
  \begin{subequations}
    \label{eq:54}
    \begin{align}
      \left<N(t)\right> &= \frac{1}{2} \left(e^{-\Gamma_1 t} +
        e^{-\Gamma_2 t}\right) (n_1 + n_2)\,,\\
      \left<S(t)\right> &= e^{-\Gamma t} \cos(\Delta m t) (n_1 - n_2)\,,
    \end{align}
  \end{subequations}
  so we get oscillations of the quantum number $S$.  It is worth
  noting, that for the particles such as $K$ or $B$ mesons we can use
  this result only as a first approximation, since for these particles
  the transformation which ``diagonalizes'' the master equation is
  non-unitary due to $CP$-violation.  For the sake of brevity, we do
  not discuss the violated $CP$-symmetry here but preliminary
  calculations show the agreement with values for masses and
  life-times of neutral $K$ or $B$ mesons estimated on the basis of
  traditional Wigner--Weisskopf approach.
\end{example}

\section{Conclusions}
\label{sec:concl}
\noindent
We have analyzed a class of master equations built up from creation
and annihilation operators which generate dynamical semigroups that
can describe the exponential decay and flavour oscillations for system
of many particles.  We have shown, in this case, how this dynamical
semigroup can be written in the Schr\"odinger as well as Heisenberg
picture.  This allowed us to choose the picture which seems to be more
convenient for the description of the system under consideration.
Moreover, we have found the solution for a free particle master
equation in the form of Kraus representation in the language of
annihilation and creation operators.  Although, this Kraus
representation is given by an infinite series, in the Schr\"odinger
picture it reduces to a finite sum, whenever the initial state has a
finite number of particles.  On the other hand, in the Heisenberg
picture the commutation relations between observables and Kraus
operators sometimes allows us to find the observable evolution in
closed form without explicit summation of the series.

Notice that if we cut the presented approach to the one-zero particle
sector we get the theory given in \cite{caban05a} (neglecting the
decoherence).

In the present paper we restrict our analysis only to states labeled
by a discrete index, and not by continuous parameter (like e.g.\
momentum).  Despite the fact that introducing a continuous parameter
causes creation and annihilation operators to become operator-valued
distributions, it seems to us that the approach introduced here should
also be applicable.

We left open the question whether it is possible to apply our approach
to describe the processes other than exponential decay, like e.g.\
decoherence or different decay laws.  The preliminary investigations
suggest that there exists a positive answer.

\section*{Acknowledgments}
\noindent
This work was supported by the Polish Ministry of Science and Higher
Education under Contract No.\ NN202~103738.

\bibliography{unstable2Notes}
\end{document}